\RequirePackage{fix-cm}
\documentclass[smallextended]{svjour3}       % onecolumn (second format)
\smartqed  % flush right qed marks, e.g. at end of proof
\usepackage[cmex10]{amsmath}
\usepackage{graphicx}
\usepackage{color}
\usepackage{multirow}
\usepackage{epstopdf}
\usepackage{amssymb}
\usepackage{cite}
\usepackage{bbm}
\newtheorem{lem}{Lemma}

\newtheorem{prop}{Proposition}

\begin{document}

\title{All-Optical FSO Relaying Under Mixture-Gamma Fading Channels and Pointing Errors%\thanks{Grants or other notes
%about the article that should go on the front page should be
%placed here. General acknowledgments should be placed at the end of the article.}
}

%\titlerunning{Short form of title}        % if too long for running head

\author{Nikolaos~I.~Miridakis         \and
       Dimitrios~D.~Vergados \and
			Angelos~Michalas%etc.
}

%\authorrunning{Short form of author list} % if too long for running head

\institute{Nikolaos~I.~Miridakis \at
              School of Electrical and Information Engineering and the Institute of Physical Internet, Jinan University, Zhuhai 519070, China and Department of Computer Systems Engineering, Piraeus University of Applied Sciences, 12244 Aegaleo, Greece \\
                            \email{nikozm@unipi.gr}           %  \\
%             \emph{Present address:} of F. Author  %  if needed
           \and
           Dimitrios~D.~Vergados \at
              Department of Informatics, University of Piraeus, 185 34, Piraeus, Greece
							\and
							Angelos~Michalas \at 
							Department of Informatics and Computer Technology, Technological Education Institute of Western Macedonia, 52 100, Kastoria, Greece
}

\date{Received: date / Accepted: date}
% The correct dates will be entered by the editor

\maketitle

\begin{abstract}
The performance of all-optical dual-hop relayed free-space optical communication systems is analytically studied and evaluated. We consider the case when the total received signal undergoes turbulence-induced channel fading, modeled by the versatile mixture-Gamma distribution. Also, the misalignment-induced fading due to the presence of pointing errors is jointly considered in the enclosed analysis. The performance of both amplify-and-forward and decode-and-forward relaying transmission is studied, when heterodyne detection is applied. New closed-form expressions are derived regarding some key performance metrics of the considered system; namely, the system outage probability and average bit-error rate. 

% \PACS{PACS code1 \and PACS code2 \and more}
% \subclass{MSC code1 \and MSC code2 \and more}
\end{abstract}

\section{Introduction}
\label{intro}
One of the main challenges in free-space optical (FSO) communications is the channel fading due to turbulence-induced scintillation and misalignment. The former is mainly caused by the random fluctuation of the refractive index, whereas the latter by time-varying pointing errors due to thermal expansion, dynamic wind loads, and internal vibrations \cite{j:FaridHranilovic2012}. To mitigate channel fading, an all-optical cooperative transmission via relaying can be used, which brings impressive performance improvements in FSO systems \cite{j:KashaniUysal2013,j:ChatzidiamantisMichalopoulos2013,j:YangGaoAlouini2014}.

To efficiently address the composite turbulence-induced and misalignment fading, various distribution models have been proposed. Among them, the most popular ones (due to their accuracy and generality) are the Gamma-Gamma and M\'alaga distribution models. These models, combined with the pointing error effect, result to highly complicated probability density function (PDF)/cumulative distribution function (CDF) because the rather cumbersome high-order Meijer's $G$-function is included therein \cite{j:AnsariYilmaz2016}. Due to this reason, the performance of relayed transmission has been only merely studied so far in FSO communication systems. Specifically, only approximate analytical results have been presented to date, regarding the asymptotically high signal-to-noise ratio (SNR) regime (e.g., see \cite{j:AnsariYilmaz2016,j:YangGaoAlouini2014} and relevant references therein). Nevertheless, analytical results for the performance of all-optical relayed FSO systems in the entire SNR regime (low-to-high) are not available so far.

Capitalizing on the aforementioned observations, we study the performance of all optical dual-hop relayed FSO systems over composite fading channels and pointing errors in arbitrary SNR regions. Three popular relaying transmission protocols are adopted; namely, the channel state information (CSI)-assisted amplify-and-forward (AF), the AF with a fixed (statistical) gain, and the decode-and-forward (DF) relaying transmission schemes. The mixture Gamma (MG) distribution model is used to capture the statistical properties of turbulence-induced channel fading. It was recently indicated in \cite{j:SandalidisMG2016} that the MG distribution sharply coincides to both the Gamma-Gamma and M\'alaga generalized distribution models. Thus, it serves as an effective model to accurately approach channel fading in weak-to-strong turbulence conditions. Further, the misalignment-induced fading is also considered, by assuming the commonly adopted model of zero-boresight pointing errors. 

The analysis of current work focuses on the coherent heterodyne detection type. Although it represents a more complicated detection method than other alternatives, it has the ability to better overcome the thermal noise effects \cite{j:YangGaoAlouini2014,j:LiuYao2010}. New closed-form expressions are derived for key system performance metrics, i.e., the system outage probability and ABER. The derived results can serve as an efficient means of evaluating the performance of the considered relaying schemes, whereas they are much more time-effective than the available methods so far (i.e., time-consuming manifold numerical integrations and/or exhaustive simulations). 

\emph{Notation}: $f_{X}(\cdot)$, $F_{X}(\cdot)$ and $\bar{F}_{X}(\cdot)$ represent PDF, CDF and complementary CDF of the random variable (RV) $X$, respectively. Moreover, $\Gamma(\cdot)$ denotes the Gamma function \cite[Eq. (8.310.1)]{tables}, $\Gamma(\cdot,\cdot)$ is the upper incomplete Gamma function \cite[Eq. (8.350.2)]{tables}, and ${\rm erf(\cdot)}$ represents the error function \cite[Eq. (8.253.1)]{tables}. Finally, ${}_2F_1(\cdot,\cdot;\cdot;\cdot)$ denotes the Gauss hypergeometric function \cite[Eq. (9.100)]{tables}, while $W_{\mu,\nu}(\cdot)$ is the Whittaker's-$W$ hypergeometric function \cite[Eq. (9.220.4)]{tables}.

\section{System Model}
\label{System Model}
Consider an all-optical dual-hop relayed FSO system, which consists of a source, relay and destination node; all equipped with a single aperture. Two transmission phases (e.g., two consecutive time slots) are required to complete the end-to-end information transmission. In the first transmission phase, the source converts the RF signal to optical signal and send it to the relay node. Upon reception, the relay process the latter signal according to the underlying relaying transmission protocol. Then, in the second transmission phase, the relay retransmits the signal to the destination. It is assumed that the direct source-to-destination link does not exist, due to strong propagation attenuation and/or severe channel fading conditions. We retain our focus on three popular relaying protocols; namely, CSI-assisted AF relaying, AF relaying using a fixed (statistical) gain, and DF relaying. Correspondingly, the end-to-end SNR	is given by \cite{j:YangGaoAlouini2014}
\begin{align}
\nonumber
\left\{
\begin{array}{c l} 
\gamma_{\rm CSI}&\triangleq \frac{\gamma_{1}\gamma_{2}}{\gamma_{1}+\gamma_{2}+q}, \\
\gamma_{\rm Fixed}&\triangleq \frac{\gamma_{1}\gamma_{2}}{\gamma_{2}+U},\\
\gamma_{\rm DF}&\triangleq \min\{\gamma_{1},\gamma_{2}\},
\end{array}\right.
\end{align}  
where $\gamma_{1}$ and $\gamma_{2}$ denote the received SNR at the two consecutive transmission phases, while $U$ is a fixed constant related to the fixed AF gain. Also, the parameter $q\in\{0,1\}$ corresponds to the approximate CSI-assisted AF relayed case (when $q=0$) or the exact case (when $q=1$).

The received signal at the two transmission phases undergoes joint channel fading (irradiance) and misalignment due to the possible pointing errors, whereas it is subject to an additive white zero-mean Gaussian noise with power $N_{0}$. The recently proposed MG distribution is used to model the turbulence-induced fading. Doing so, it turns out that the PDF of SNR at the $j^{\rm th}$ transmission hop, with $j\in \{1,2\}$, reads as \cite{j:SandalidisMG2016}
\begin{align}
\nonumber
f_{\gamma_{j}}(x)=&\sum^{L}_{i_{j}=1}\frac{a_{i_{j}}c^{\xi^{2}_{j}-b_{i_{j}}}_{i_{j}}\xi^{2}_{j}x^{\xi^{2}_{j}-1}}{(A_{0_{j}}\bar{\gamma}_{j})^{\xi^{2}_{j}}}\\
&\times \Gamma\left(b_{i_{j}}-\xi^{2}_{j},\frac{c_{i_{j}}}{A_{0_{j}}\bar{\gamma}_{j}}x\right),\ \ x>0,
\label{snrpdf}
\end{align}   
where $\bar{\gamma}_{1}\triangleq P_{S}\eta_{1}\bar{I}_{1}/N_{0}$ and $\bar{\gamma}_{2}\triangleq P_{R}\eta_{2}\bar{I}_{2}/N_{0}$ denote the average received SNRs at the relay and destination nodes, respectively, while $\eta_{j}$ is the electrical-to-optical conversion coefficient of the $j^{\rm th}$ hop.\footnote{In the AF relaying case, where there is no optical-to-electrical conversion at the relay node, $\eta_{2}=1$ \cite{j:YangGaoAlouini2014}.} Also, $P_{S}$ and $P_{R}$ represent the transmit power of the source and relay, respectively, while $\bar{I}_{j}$ denotes the average received irradiance at the $j^{\rm th}$ hop, which is expressed as \cite{j:SandalidisMG2016}
\begin{align}
\bar{I}_{j}=\left(\frac{\xi^{2}_{j}A_{0_{j}}}{1+\xi^{2}_{j}}\right)\sum^{L}_{i_{j}=1}a_{i_{j}}\Gamma(1+b_{i_{j}})c^{-(1+b_{i_{j}})}_{i_{j}}.
\label{meanI}
\end{align} 
Moreover, $a_{i_{j}}$, $b_{i_{j}}$ and $c_{i_{j}}$ denote the parameters of MG distribution. As an illustrative example, to model Gamma-Gamma channel fading, these parameters are defined in \cite[Eq. (5)]{j:SandalidisMG2016}. A similar matching can be easily obtained when considering another generalized fading model, namely, the M\'alaga distribution \cite[Eq. (9)]{j:SandalidisMG2016}. At this point, it is worth noting that $b_{i_{j}}$ reflects the minimum between the small- and large scale channel fading parameters at the $j^{\rm th}$ hop \cite{j:MiridakisTsiftsisletter2017,j:MiridakisVergadosMichalllllas2017}. The total number of sum terms in \eqref{snrpdf} determines the accuracy level of MG distribution. It was shown in \cite[Figs. 2 and 3]{j:SandalidisMG2016} that a condition of $L\leq 10$ satisfies an acceptable accuracy level for most practical applications, resulting to a rapidly converging series of \eqref{snrpdf}. This reveals the great success of the MG distribution model. Further, $A_{0}$ is a constant term that defines the pointing loss given by $A_{0}=[{\rm erf}(\sqrt{\pi}r/(\sqrt{2}w_{z}))]^{2}$, while $r$ is the radius of the detection aperture, and $w_{z}$ is the beam waist. In addition, $\xi$ denotes the ratio between the equivalent beam radius at the receiver and the pointing error displacement standard deviation (jitter) at the receiver (i.e., the condition when $\xi\rightarrow +\infty$ reflects the non-pointing error case).

\section{Statistical Analysis}
\label{Statistical Analysis}
We commence by extracting some key statistical results, namely, the CDF of the end-to-end received SNR for various relaying transmission modes. These results are quite useful for the overall performance evaluation of the considered system.

Under the condition $b_{i_{j}}-\xi^{2}_{j} \in \mathbb{N}^{+}$ and using \cite[Eq. (8.352.4)]{tables}, \eqref{snrpdf} reduces to
\begin{align}
f_{\gamma_{j}}(x)=\sum^{L}_{i_{j}=1}\sum^{b_{i_{j}}-\xi^{2}_{j}-1}_{k_{j}=0}\Xi(i_{j})x^{\xi^{2}_{j}+k_{j}-1}\exp\left(-\frac{c_{i_{j}}}{A_{0_{j}}\bar{\gamma}_{j}}x\right),
\label{snrpdf1}
\end{align} 
where
\begin{align}
\Xi(i_{j})\triangleq \frac{a_{i_{j}}c^{\xi^{2}_{j}-b_{i_{j}}+k_{j}}_{i_{j}}\xi^{2}_{j}(b_{i_{j}}-\xi^{2}_{j}-1)!}{k_{j}!(A_{0_{j}}\bar{\gamma}_{j})^{\xi^{2}_{j}+k_{j}}}.
\label{Xi}
\end{align} 

Notice that the restriction $b_{i_{j}}-\xi^{2}_{j} \in \mathbb{N}^{+}$ applies for integer channel fading parameters (or to be more precise, $b_{i_{j}}-\xi^{2}_{j}$ should be an integer value). In the case when at least one of the involved parameters is non-integer, \eqref{snrpdf1} can serve as a close-approximate of \eqref{snrpdf} and/or as a performance benchmark, since it bounds the SNR performance between the nearest upper and lower integer values of the corresponding fading parameter.  

Moreover, in the case when $b_{i_{j}}-\xi^{2}_{j}\leq 0$ (i.e., the misalignment-induced parameter due to the pointing error effect is equal or greater than the turbulence-induced fading parameter), then the following inequality is used
\begin{align}
\nonumber
\Gamma(-j,x)&\triangleq \int^{\infty}_{x}t^{-j-1}\exp(-t)dt\\
\nonumber
&=\exp(-x)\int^{\infty}_{0}(t+x)^{-j-1}\exp(-t)dt\\
&\leq \exp(-x)x^{-j-1},\ j\geq 0,
\label{inequality}
\end{align} 
which gets more tight for low-to-moderate SNR regions.\footnote{The sharpness of this bound is manifested when $x\rightarrow +\infty$, such that $\Gamma(-j,x)\rightarrow \exp(-x)x^{-j-1}$, which reflects low-to-moderate SNR regions according to \eqref{snrpdf}. It is noteworthy that the latter SNR regions are of particular interest in current study, since analytical results only for the asymptotically high SNR region exist so far.} According to the latter inequality, it stems that
\begin{align}
f_{\gamma_{j}}(x)\leq f^{(B)}_{\gamma_{j}}(x),\ x>0,
\label{snrpdfbound}
\end{align}  
where 
\begin{align}
f^{(B)}_{\gamma_{j}}(x)=\sum^{L}_{i_{j}=1}\frac{a_{i_{j}}c^{-1}_{i_{j}}\xi^{2}_{j}}{(A_{0_{j}}\bar{\gamma}_{j})^{b_{i_{j}}-1}}x^{b_{i_{j}}-2}\exp\left(-\frac{c_{i_{j}}}{A_{0_{j}}\bar{\gamma}_{j}}x\right),
\end{align} 
stands for the upper bound of the true PDF $f_{\gamma_{j}}(\cdot)$ and can be directly obtained from \eqref{snrpdf1} by simply setting $k_{j}=b_{i_{j}}-\xi^{2}_{j}-1$. In what follows, we use the general form of \eqref{snrpdf1} and we revisit the tightness of the bound (in the case when $b_{i_{j}}-\xi^{2}_{j}\leq 0$) in the numerical results. 

The CCDF of SNR at the $j^{\rm th}$ transmission hop is given by
\begin{align}
\nonumber
\bar{F}_{\gamma_{j}}(x)&=1-F_{\gamma_{j}}(x)\\
\nonumber
&=\sum^{L}_{i_{j}=1}\sum^{b_{i_{j}}-\xi^{2}_{j}-1}_{k_{j}=0}\sum^{\xi^{2}_{j}+k_{j}-1}_{r_{j}=0}\frac{\Xi(i_{j})(\xi^{2}_{j}+k_{j}-1)!}{r_{j}!\left(\frac{c_{i_{j}}}{A_{0_{j}}\bar{\gamma}_{j}}\right)^{\xi^{2}_{j}+k_{j}-r_{j}}}\\
&\ \ \ \ \times x^{r_{j}} \exp\left(-\frac{c_{i_{j}}}{A_{0_{j}}\bar{\gamma}_{j}}x\right),
\label{ccdfsnr}
\end{align} 

Based on the above statistics, the following lemmas provide the CDF for end-to-end SNR for different types of relayed transmission mode.

\begin{lem}
The CDF of the end-to-end SNR for a dual-hop CSI-assisted relayed transmission is derived by \eqref{cdfafcsi} in a closed-form expression.

\begin{align}
\nonumber
F_{\gamma_{\rm CSI}}(x)&=1-\sum^{L}_{i_{1}=1}\sum^{b_{i_{1}}-\xi^{2}_{1}-1}_{k_{1}=0}\sum^{\xi^{2}_{1}+k_{1}-1}_{r_{1}=0}\sum^{L}_{i_{2}=1}\sum^{b_{i_{2}}-\xi^{2}_{2}-1}_{k_{2}=0}\sum^{\xi^{2}_{2}+k_{2}-1}_{p=0}\sum^{r_{1}}_{s=0}\binom{r_{1}}{s}\binom{\xi^{2}_{2}+k_{2}-1}{p}\\
\nonumber
&\times \frac{2\Xi(i_{1})\Xi(i_{2})(\xi^{2}_{1}+k_{1}-1)!\left(\frac{c_{i_{1}}}{A_{0_{1}}\bar{\gamma}_{1}}\right)^{\frac{p-s+1}{2}+r_{1}-\xi^{2}_{1}-k_{1}}(x^{2}+q x)^{\frac{p+s+1}{2}}}{r_{1}!\left(\frac{c_{i_{2}}}{A_{0_{2}}\bar{\gamma}_{2}}\right)^{\frac{p-s+1}{2}}\exp\left(\left[\frac{c_{i_{1}}}{A_{0_{1}}\bar{\gamma}_{1}}+\frac{c_{i_{2}}}{A_{0_{2}}\bar{\gamma}_{2}}\right]x\right)}\\
&\times x^{\xi^{2}_{2}+k_{2}+r_{1}-p-s-1} K_{p-s+1}\left(2\sqrt{\frac{c_{i_{1}}c_{i_{2}}}{A_{0_{1}}A_{0_{2}}\bar{\gamma}_{1}\bar{\gamma}_{2}}(x^{2}+q x)}\right)
\label{cdfafcsi}
\end{align}
\end{lem}

\begin{proof}
It holds that \cite{tsiftsis2006nonregenerative} 
\begin{align}
F_{\gamma_{\rm CSI}}(x)\triangleq 1-\int^{\infty}_{0}\bar{F}_{\gamma_{1}}\left(x+\frac{x^{2}+q x}{y}\right)f_{\gamma_{2}}(x+y)dy.
\label{cdfcsiii}
\end{align}
Hence, inserting \eqref{ccdfsnr} and \eqref{snrpdf1} into \eqref{cdfcsiii}, utilizing \cite[Eqs. (1.111) and (3.471.9)]{tables}, and after performing some straightforward manipulations, (\ref{cdfafcsi}) is obtained.
\end{proof}

\begin{lem}
The CDF of the end-to-end SNR for a dual-hop relayed transmission with a fixed gain is provided as
\begin{align}
\nonumber
F_{\gamma_{\rm Fixed}}(x)&=1-\sum^{L}_{i_{1}=1}\sum^{b_{i_{1}}-\xi^{2}_{1}-1}_{k_{1}=0}\sum^{\xi^{2}_{1}+k_{1}-1}_{r_{1}=0}\sum^{L}_{i_{2}=1}\sum^{b_{i_{2}}-\xi^{2}_{2}-1}_{k_{2}=0}\sum^{r_{1}}_{s=0}\binom{r_{1}}{s}\Xi(i_{1})\Xi(i_{2})\\
\nonumber
&\times \frac{2(\xi^{2}_{1}+k_{1}-1)!\left(\frac{c_{i_{1}}}{A_{0_{1}}\bar{\gamma}_{1}}\right)^{\frac{\xi^{2}_{2}+k_{2}-s}{2}+r_{1}-\xi^{2}_{1}-k_{1}}U^{\frac{\xi^{2}_{2}+k_{2}+s}{2}}x^{\frac{\xi^{2}_{2}+k_{2}-s}{2}+r_{1}}}{r_{1}!\left(\frac{c_{i_{2}}}{A_{0_{2}}\bar{\gamma}_{2}}\right)^{\frac{\xi^{2}_{2}+k_{2}-s}{2}}\exp\left(\frac{c_{i_{1}}}{A_{0_{1}}\bar{\gamma}_{1}}x\right)}\\
&\times K_{\xi^{2}_{2}+k_{2}-s}\left(2\sqrt{\frac{c_{i_{1}}c_{i_{2}}}{A_{0_{1}}A_{0_{2}}\bar{\gamma}_{1}\bar{\gamma}_{2}}U x}\right),
\label{cdfaffixed}
\end{align}
where 
\begin{align}
\nonumber
U\triangleq \left(\mathbb{E}\left[\frac{1}{\gamma_{1}+1}\right]\right)^{-1}=\Bigg[&\sum^{L}_{i_{1}=1}\sum^{b_{i_{1}}-\xi^{2}_{1}-1}_{k_{1}=0}\frac{\Xi(i_{1})\Gamma(\xi^{2}_{1}+k_{1})}{\exp\left(-\frac{c_{i_{1}}}{A_{0_{1}}\bar{\gamma}_{1}}\right)}\\
&\times \Gamma\left(1-\xi^{2}_{1}-k_{1},\frac{c_{i_{1}}}{A_{0_{1}}\bar{\gamma}_{1}}\right)\Bigg]^{-1}.
\label{ugain}
\end{align} 
\end{lem}

\begin{proof}
It holds that \cite{tsiftsis2006nonregenerative}
\begin{align}
F_{\gamma_{\rm Fixed}}(x)\triangleq \int^{\infty}_{0}F_{\gamma_{1}}\left(x+\frac{U x}{y}\right)f_{\gamma_{2}}(y)dy,
\label{cdffixed}
\end{align}
Hence, following the same steps as for deriving (\ref{cdfafcsi}), yields (\ref{cdfaffixed}). Also, regarding the derivation of \eqref{ugain}, the fixed AF gain is defined as \cite{j:HasnaAlouini2004} $U\triangleq (\int^{\infty}_{0}(x+1)^{-1}f_{\gamma_{1}}(x)dx)^{-1}$. Hence, using \cite[Eq. (3.383.10)]{tables}, \eqref{ugain} is directly extracted.
\end{proof}

\begin{lem}
The CDF of the end-to-end SNR for a DF relayed transmission is expressed as
\begin{align}
F_{\gamma_{\rm DF}}(x)\triangleq F_{\min\{\gamma_{1},\gamma_{2}\}}(x)=1-\prod^{2}_{j=1}\bar{F}_{\gamma_{j}}(x).
\label{cdfdf}
\end{align} 
\end{lem}

\section{System Performance}
\label{System Performance}

Capitalizing on the previously derived statistics, some useful metrics that define the overall system performance are analytically evaluated.

\subsection{Outage Probability}
The outage probability is defined as the probability that the SNR falls below a certain threshold value, $\gamma_{\text{th}}$, such that $P_{\rm out}(\gamma_{\rm th})=F_{\gamma}(\gamma_{\rm th})$. With the aid of the closed-form expressions of the previously derived lemmas, the system outage performance can be easily computed for each relayed transmission mode, correspondingly.

\subsection{Average Bit-Error Rate}
The ABER is defined as \cite{j:AnsariYilmaz2016}
\begin{align}
\overline{P}_{b}\triangleq \frac{Q^{P}}{2 \Gamma(P)}\int^{\infty}_{0}z^{P-1}\exp\left(-Q z\right)F_{\gamma}(z)dz,
\label{aberdef}
\end{align}
where $P$ and $Q$ are fixed modulation-specific parameters. For the CSI-assisted AF relaying scenario, the exact ABER is analytically intractable. Yet, the case when $q=0$ can serve as a sharp approximation in moderate-to-high SNR.

\begin{prop}
The ABER of the end-to-end SNR for a dual-hop CSI-assisted relayed transmission, when $q=0$, is presented as
\begin{align}
\nonumber
&\bar{P_{b}}_{\gamma_{\rm CSI}}=\scriptstyle \frac{1}{2}-\sum^{L}_{i_{1}=1}\sum^{b_{i_{1}}-\xi^{2}_{1}-1}_{k_{1}=0}\sum^{\xi^{2}_{1}+k_{1}-1}_{r_{1}=0}\sum^{L}_{i_{2}=1}\sum^{b_{i_{2}}-\xi^{2}_{2}-1}_{k_{2}=0}\sum^{\xi^{2}_{2}+k_{2}-1}_{p=0}\sum^{r_{1}}_{s=0}\binom{r_{1}}{s}\binom{\xi^{2}_{2}+k_{2}-1}{p}\\
\nonumber
&\times \scriptstyle  \frac{\Xi(i_{1})\Xi(i_{2})(\xi^{2}_{1}+k_{1}-1)!\sqrt{\pi}4^{p-s+1}\left(\frac{c_{i_{1}}}{A_{0_{1}}\bar{\gamma}_{1}}\right)^{p-s+r_{1}-\xi^{2}_{1}-k_{1}+1}\Gamma(\xi^{2}_{2}+k_{2}+r_{1}+p-s+P+1)\Gamma(\xi^{2}_{2}+k_{2}+r_{1}-p+s+P-1)}{r_{1}!\left(\frac{c_{i_{1}}}{A_{0_{1}}\bar{\gamma}_{1}}+\frac{c_{i_{2}}}{A_{0_{2}}\bar{\gamma}_{2}}+Q+2\sqrt{\frac{c_{i_{1}}c_{i_{2}}}{A_{0_{1}}A_{0_{2}}\bar{\gamma}_{1}\bar{\gamma}_{2}}}\right)^{\xi^{2}_{2}+k_{2}+r_{1}+p-s+P+1}\Gamma\left(\xi^{2}_{2}+k_{2}+r_{1}+P+\frac{1}{2}\right)}\\
&\times \scriptstyle {}_2F_1\left(\xi^{2}_{2}+k_{2}+r_{1}+p-s+P+1,p-s+\frac{3}{2};\xi^{2}_{2}+k_{2}+r_{1}+P+\frac{1}{2};\frac{\frac{c_{i_{1}}}{A_{0_{1}}\bar{\gamma}_{1}}+\frac{c_{i_{2}}}{A_{0_{2}}\bar{\gamma}_{2}}+Q-2\sqrt{\frac{c_{i_{1}}c_{i_{2}}}{A_{0_{1}}A_{0_{2}}\bar{\gamma}_{1}\bar{\gamma}_{2}}}}{\frac{c_{i_{1}}}{A_{0_{1}}\bar{\gamma}_{1}}+\frac{c_{i_{2}}}{A_{0_{2}}\bar{\gamma}_{2}}+Q+2\sqrt{\frac{c_{i_{1}}c_{i_{2}}}{A_{0_{1}}A_{0_{2}}\bar{\gamma}_{1}\bar{\gamma}_{2}}}}\right)
\label{aberafcsi}
\end{align}
\end{prop}

\begin{proof}
The desired expression directly arises by setting $q=0$ in \eqref{cdfafcsi}, using \eqref{aberdef}, and utilizing \cite[Eq. (6.621.3)]{tables}. 
\end{proof}

\begin{prop}
The ABER of the end-to-end SNR for a dual-hop relayed transmission with a fixed gain is derived in a closed-form expression as 
\begin{align}
\nonumber
&\bar{P_{b}}_{\gamma_{\rm Fixed}}=\scriptstyle \frac{1}{2}-\sum^{L}_{i_{1}=1}\sum^{b_{i_{1}}-\xi^{2}_{1}-1}_{k_{1}=0}\sum^{\xi^{2}_{1}+k_{1}-1}_{r_{1}=0}\sum^{L}_{i_{2}=1}\sum^{b_{i_{2}}-\xi^{2}_{2}-1}_{k_{2}=0}\sum^{r_{1}}_{s=0}\binom{r_{1}}{s}\Xi(i_{1})\Xi(i_{2})\\
\nonumber
&\times \scriptstyle \frac{(\xi^{2}_{1}+k_{1}-1)!\left(\frac{c_{i_{1}}}{A_{0_{1}}\bar{\gamma}_{1}}\right)^{\frac{\xi^{2}_{2}+k_{2}-s-1}{2}+r_{1}-\xi^{2}_{1}-k_{1}}U^{\frac{\xi^{2}_{2}+k_{2}+s}{2}}Q^{P}\Gamma(\xi^{2}_{2}+k_{2}-s+r_{1}+P)\Gamma(r_{1}+P)}{2 r_{1}!\Gamma(P)\left(\frac{c_{i_{2}}}{A_{0_{2}}\bar{\gamma}_{2}}\right)^{\frac{\xi^{2}_{2}+k_{2}-s+1}{2}}\left(\frac{c_{i_{1}}}{A_{0_{1}}\bar{\gamma}_{1}}+Q\right)^{\frac{\xi^{2}_{2}+k_{2}-s+2r_{1}+2P-1}{2}}\exp\left(-\frac{c_{i_{1}}c_{i_{2}}}{2 A_{0_{1}}A_{0_{2}}\bar{\gamma}_{1}\bar{\gamma}_{2}\left(\frac{c_{i_{1}}}{A_{0_{1}}\bar{\gamma}_{1}}+Q\right)}\right)}\\
&\times \scriptstyle W_{-\frac{(\xi^{2}_{2}+k_{2}-s+2 r_{1}+2 P-1)}{2},\frac{\xi^{2}_{2}+k_{2}-s}{2}}\left(\frac{c_{i_{1}}c_{i_{2}}}{A_{0_{1}}A_{0_{2}}\bar{\gamma}_{1}\bar{\gamma}_{2}\left(\frac{c_{i_{1}}}{A_{0_{1}}\bar{\gamma}_{1}}+Q\right)}\right)
\label{aberaffixed}
\end{align}
\end{prop}

\begin{proof}
The closed-form ABER directly arises by inserting \eqref{cdfaffixed} into \eqref{aberdef}, implementing a change of variables $\sqrt{x}\rightarrow x$ in the resultant expression and utilizing \cite[Eq. (6.631.3)]{tables}. 
\end{proof}

\begin{prop}
The ABER of the end-to-end SNR for a dual-hop DF relayed transmission is presented as
\begin{align}
\nonumber
\bar{P_{b}}_{\gamma_{\rm DF}}=&\frac{1}{2}-\sum^{L}_{i_{1}=1}\sum^{b_{i_{1}}-\xi^{2}_{1}-1}_{k_{1}=0}\sum^{\xi^{2}_{1}+k_{1}-1}_{r_{1}=0}\sum^{L}_{i_{2}=1}\sum^{b_{i_{2}}-\xi^{2}_{2}-1}_{k_{2}=0}\sum^{\xi^{2}_{2}+k_{2}-1}_{r_{2}=0}\prod^{2}_{j=1}\left[\frac{\Xi(i_{j})(\xi^{2}_{j}+k_{j}-1)!}{r_{j}!\left(\frac{c_{i_{j}}}{A_{0_{j}}\bar{\gamma}_{j}}\right)^{\xi^{2}_{j}+k_{j}-r_{j}}}\right] \\
&\times \frac{Q^{P}\Gamma(r_{1}+r_{2}+P)}{2 \Gamma(P)\left(\frac{c_{i_{1}}}{A_{0_{1}}\bar{\gamma}_{1}}+\frac{c_{i_{2}}}{A_{0_{2}}\bar{\gamma}_{2}}+Q\right)^{r_{1}+r_{2}+P}}
\label{aberdf}
\end{align} 
\end{prop}

\begin{proof}
The desired result trivially follows by evaluating the integral that arises by inserting \eqref{cdfdf} into \eqref{aberdef}. 
\end{proof}

\section{Numerical Results}
\label{Numerical Results}
In this Section, the accuracy of the proposed approach is numerically verified. For the sake of clarity and without loss of generality, the analytical results are cross-compared with numerical results modeled by simulating Gamma-Gamma faded channels. Hence, $\alpha_{j}$ and $\beta_{j}$ denote the small- and large-scale channel fading parameters of the $j^{\rm th}$ transmission hop, respectively. Recall that $b_{i_{j}}=\min\{\alpha_{j},\beta_{j}\}$. Also, $L=10$, $r/w_{z}=0.1$, and $[P,Q]=[0.5,1]$ (i.e., the coherent BPSK modulation scheme is considered). The considered outage threshold is $\gamma_{\rm th}=0$dB (i.e., ensuring a minimum target data rate of 1 bps/Hz). Curve-lines and circle-marks stand for the analytical and simulation results, respectively.

In Figs.~\ref{fig1} and~\ref{fig2}, the outage and ABER performance is, respectively, illustrated for the considered dual-hop relay schemes. The CSI-assisted AF relaying outperforms the AF scheme with a fixed gain, while DF outperforms both AF schemes, as expected. Insightfully, the impact of the turbulence-induced fading severity is more critical to the system performance than the pointing error effect. Furthermore, when $\min\{\alpha_{j},\beta_{j}\}> \xi^{2}_{j}$, the analytical results perfectly match the corresponding simulation points in the entire SNR regime. In the case when $\min\{\alpha_{j},\beta_{j}\}\leq \xi^{2}_{j}$, the analytical results closely bound the simulations, while the tightness of this bound is more emphatic for low-to-moderate SNR regions. Thus, the remark of \eqref{inequality} is verified. It is also noteworthy that the diversity order is the same for all the relay schemes and is related to the minimum of the involved channel fading parameters \cite{j:FaridHranilovic2012,j:YangGaoAlouini2014}.  

\begin{figure}[!t]
\centering
\includegraphics[trim=1.5cm 0.3cm 2.5cm 1.2cm, clip=true,totalheight=0.4\textheight]{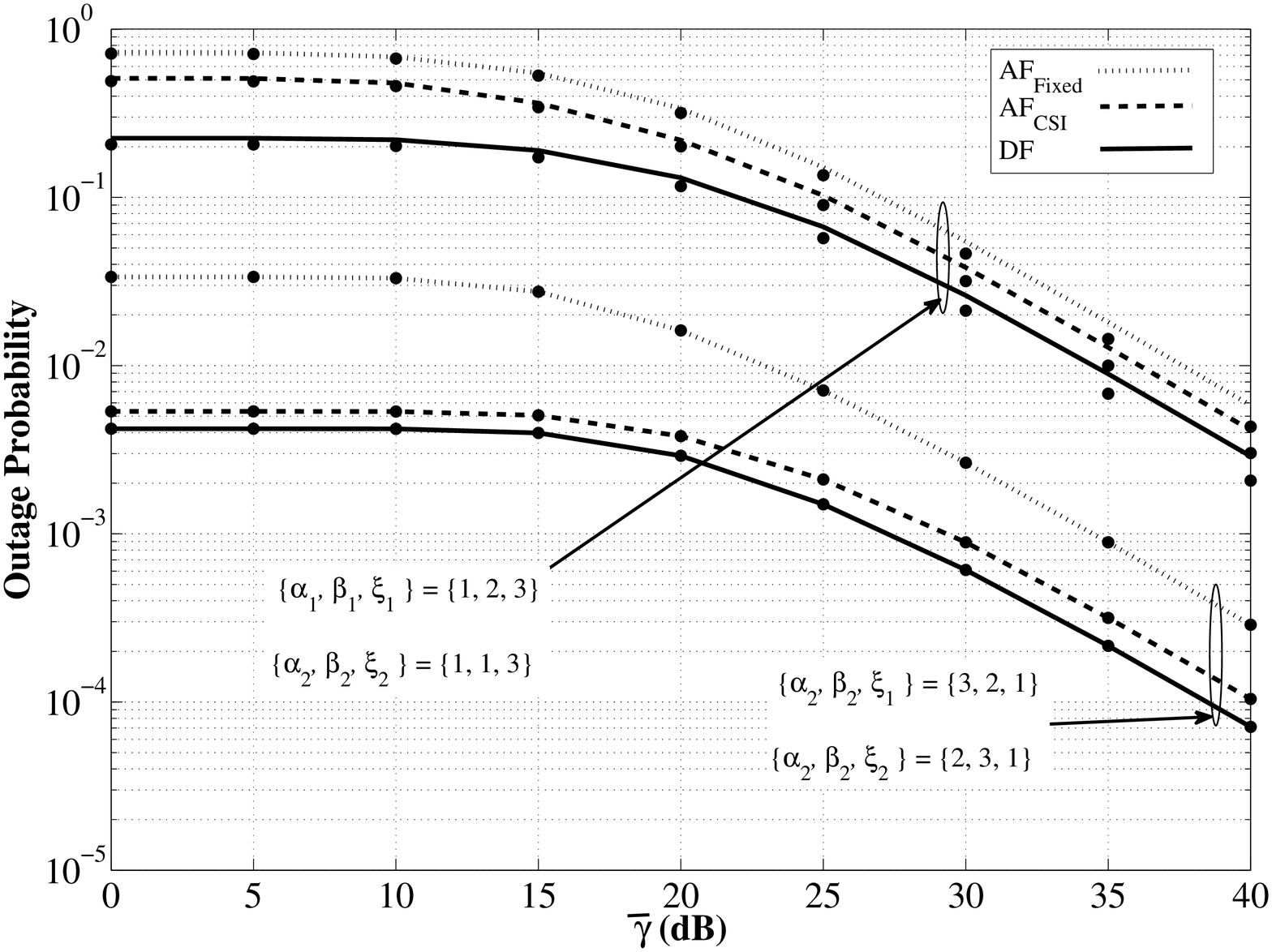}
\caption{Outage probability vs. various average received SNR values, where $\bar{\gamma}_{1}=\bar{\gamma}_{2}\triangleq \bar{\gamma}$.}
\label{fig1}
\end{figure}

\begin{figure}[!t]
\centering
\includegraphics[trim=1.5cm 0.3cm 2.5cm 1.2cm, clip=true,totalheight=0.4\textheight]{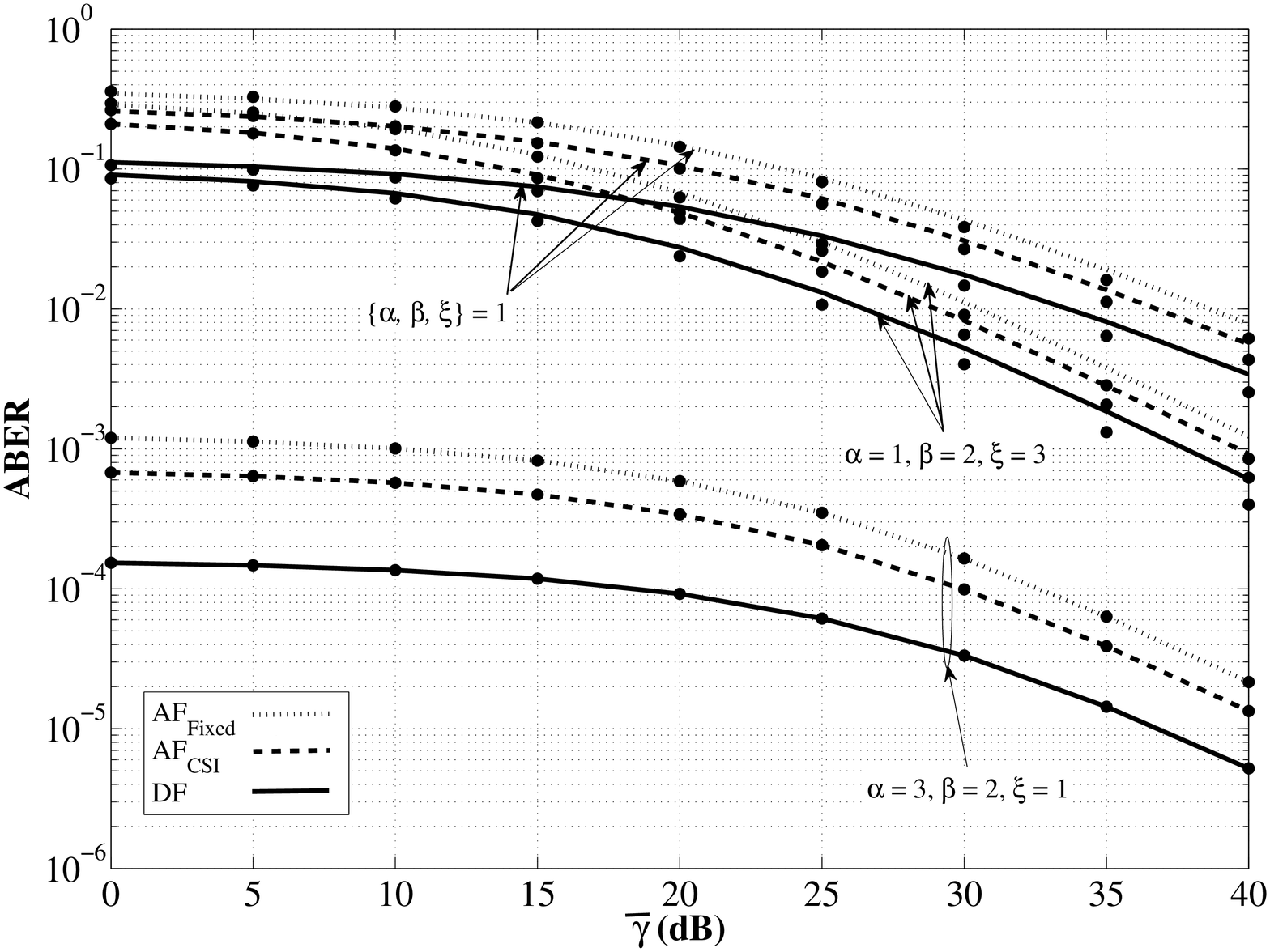}
\caption{ABER vs. various average received SNR values, under identical channel fading parameters for each transmission hop.}
\label{fig2}
\end{figure}

\section{Conclusion}
\label{Conclusion}
The performance of all-optical dual-hop relayed FSO systems was analytically studied and evaluated in the entire SNR regime. Three popular relaying transmission protocols were adopted, i.e., the CSI-assisted AF, AF with a fixed gain and DF relaying schemes. The joint effect of turbulence- and misalignment-induced channel fading was considered. Important system performance metrics were derived in new closed formulations, i.e., the system outage probability and ABER. Finally, a numerical verification of the analytical results was implemented, manifesting the efficiency of the proposed approach. 

%\begin{acknowledgements}
%If you'd like to thank anyone, place your comments here
%and remove the percent signs.
%\end{acknowledgements}

% BibTeX users please use one of
%\bibliographystyle{spbasic}      % basic style, author-year citations
\bibliographystyle{spmpsci}      % mathematics and physical sciences
\bibliography{References}   % name your BibTeX data base

\end{document}